\documentclass[lettersize,journal]{IEEEtran}

\usepackage{amsmath,amsfonts,amsthm,amssymb}
\usepackage{algorithmic}
\usepackage{algorithm}
\usepackage{array}
\usepackage[caption=false,font=normalsize,labelfont=sf,textfont=sf]{subfig}
\usepackage{stfloats}
\usepackage{url}
\usepackage{verbatim}
\usepackage{graphicx}
\usepackage[noadjust]{cite}

\allowdisplaybreaks

\theoremstyle{remark}

\newtheorem{lemma}{\indent Lemma}
\newtheorem{corollary}{\indent Corollary}

\newtheorem{theorem}{\indent Theorem}
\newtheorem*{remark}{\indent Remark}

\begin{document}

\title{DNA-Correcting Codes in DNA Storage Systems}

\author{Huawei~Wu
  %~\IEEEmembership{}
  % <-this % stops a space
  %\thanks{K. Feng is with the Department of Mathematical Sciences, Tsinghua University, Beijing, 100084, China (e-mail: fengkq@mail.tsinghua.edu.cn)}% <-this % stops a space
  \thanks{H. Wu is with the Department of Mathematical Sciences, Tsinghua University, Beijing, 100084, China (e-mail: wu-hw18@mails.tsinghua.edu.cn)}%
}

\markboth{}%
{Shell \MakeLowercase{\textit{et al.}}: A Sample Article Using IEEEtran.cls for IEEE Journals}


\maketitle

\begin{abstract}
  In \cite{boruchovsky2023dna}, the authors proposed a new model of DNA storage system that integrates all three steps of retrieval and introduced the concept of DNA-correcting codes, which guarantees that the output of the storage system can be decoded to the original data. They also gave necessary and sufficient conditions for  DNA-correcting codes when the data part is free of errors. In this paper, we generalize their results to the general case.
\end{abstract}


\section{Introduction}
\IEEEPARstart{T}{he} potential of using synthetic DNA for large-scale information storage has been recognized for decades. A typical DNA storage system consists of three components: (1) synthesizing the strands that contain the encoded data; (2) a storage container that stores the synthetic DNA strands; (3) a DNA sequencing that reads the strands (the output sequences from the sequencing machines are called reads). In the first step, each strands has millions of copies, and the length of these strands is usually limited to 250-300 nucleotides.\\
\indent Previous works usually tackled each of the three steps independently. In \cite{boruchovsky2023dna}, the authors introduced a new solution to DNA storage that integrates all three steps of retrieval, namely clustering, reconstruction, and error correction. They also introduced the concept of DNA-correcting code, which is the unique solution to the problem of ensuring that the reads from a DNA storage system can be decoded to the original data.\\
\indent Their model is as follows. Assume that all the strands to be stored are of the same length $L$ and the alphabet is $\mathbb{F}_2=\{0,1\}$. Hence each strand can be represented as a vector in $\mathbb{F}_2^L$. As usual, a message is stored in the storage system in the form of $M$ unordered strands. To indicate the position relative to the other strands, each strand starts with a length-$l$ index-field which is distinct from those of the others and only the remaining $L-l$ bits are used to store the data. Hence each strand is of the form $s=(ind,u)$, where $ind\in\mathbb{F}_2^l$ and $u\in\mathbb{F}_2^{L-l}$, and each possible message to be stored is represented as an element of the following set:
\begin{align*}
  \mathcal{X}_{M,L,l}=\big\{\{ & (ind_1,u_1),\cdots,(ind_M,u_M)\}:                             \\
                               & \forall\ i,\ ind_i\in\mathbb{F}_2^l,u_i\in\mathbb{F}_2^{L-l}, \\
                               & \forall\ i\ne j,\ ind_i\ne ind_j\big\}.
\end{align*}
It is easy to see that $|\mathcal{X}_{M,L,l}|=\binom{2^l}{M}2^{M(L-l)}$. A subset of $\mathcal{X}_{M,L,l}$ is called a code.\\
\indent When a message $Z=\{(ind_1,u_1),\cdots,(ind_M,u_M)\}$ is synthesized, each of its strand has a large number of noisy copies (which have some error bits compared to the input strands). During the sequencing process, a subset of these copies is read. For simplicity, assume that each strand has $K$ noisy copies and thus the sequencer's output is a collection of $MK$ reads, where every output read is a noisy copy of one of the input strand. Let $\tau$ denote the maximal relative fraction of incorrect copies that every input strand can have and let $e_i$, $e_d$ denote the largest number of errors that can occur at the index-field, data-field of each strand, respectively. For simplicity, assume also that $M=2^{\beta L}$, where $0<\beta<1$ and $\beta L$ is an integer.\\
\indent Now, a $(\tau,e_i,e_d)_K$-DNA storage system can be viewed as a communication channel, which takes a message
$$Z=\{(ind_1,u_1),\cdots,(ind_M,u_M)\}\in\mathcal{X}_{M,L,l}$$
as input and randomly outputs a collection of $MK$ reads
$$\{\{(ind'_1,u'_1),\cdots,(ind'_{MK},u'_{MK})\}\}$$
such that after a relabeling if necessary,
\begin{enumerate}
  \item for any $0\le j\le M-1$, there are at most $\left\lfloor\tau K\right\rfloor$ distinct from $(ind_j,u_j)$ among $(ind_{jK+1}',u_{jK+1}')$, $\cdots$, $(ind_{jK+K}',u_{jK+K}')$,
  \item for any $0\le j\le M-1$ and any $1\le k\le K$,
        $$d_H(ind_j,ind'_{jK+k})\le e_i$$
        and
        $$d_H(u_j,u'_{jK+k})\le e_d,$$
\end{enumerate}

\noindent where $d_H$ denotes the Hamming distance.\\
\indent For any $Z\in\mathcal{X}_{M,L,l}$, let $B_{\tau,e_i,e_d}^K(Z)$ denote the set of all possible outputs of a $(\tau,e_i,e_d)_K$-DNA storage system when the input is $Z$ (i.e., any element of $B^K_{(\tau,e_i,e_d)}(Z)$ is a multiset of $MK$ reads). Under this setup, a subset $\mathcal{C}\subset\mathcal{X}_{M,L,l}$ is called a $(\tau,e_i,e_d)_K$-DNA-correcting code if for any two distinct codewords $Z$, $Z'\in\mathcal{C}$, we have $B_{(\tau,e_i,e_d)}^K(Z)\cap B^K_{(\tau,e_i,e_d)}(Z')=\emptyset$. Clearly, the requirement of a $(\tau,e_i,e_d)_K$-DNA-correcting code is equivalent to saying that any two distinct inputs cannot lead to the same output after passing through a $(\tau,e_i,e_d)_K$-DNA storage system.\\
\indent In \cite{boruchovsky2023dna}, the authors only studied the case where the data part is free of errors, i.e., $e_d=0$, giving necessary and sufficient conditions for the existence of a $(\tau,e_i,0)_K$-DNA-correcting code. In this paper, we generalize their results to the case where $e_d>0$.

\section{Conditions for DNA-correcting codes}
\indent We first recall some notations in \cite{boruchovsky2023dna}. For any
$$Z=\{(ind_1,u_1),\cdots,(ind_M,u_M)\}\in\mathcal{X}_{M,L,l},$$
$S(Z):=\{u_1,\cdots,u_M\}$ is called the data-field set of $Z$ and $MS(Z)=\{\{u_1,\cdots,u_M\}\}$ is called the data-field multiset of $Z$. The notation $MS(\mathcal{X}_{M,L,l})$ denotes the set of all possible data-field multisets of the elements in $\mathcal{X}_{M,L,l}$. For a code $\mathcal{C}\subset\mathcal{X}_{M,L,l}$ and a data-field multiset $U\subset MS(\mathcal{X}_{M,L,l})$, let $\mathcal{C}_U\subset\mathcal{C}$ be the set of all codewords $Z\subset \mathcal{C}$ such that $MS(Z)=U$.\\
\indent If $A$ and $B$ are two sets of the same cardinality, let ${\rm{BI}}(A,B)$ be the set of all bijections from $A$ to $B$. For any $Z=\{(ind_1,u_1),\cdots,(ind_M,u_M)\}\in\mathcal{X}_{M,L,l}$ and $u\in S(Z)$, let $I(u,Z)$ be the set of all indices of $u$ in $Z$, that is, $I(u,Z)=\{{\rm{ind}}_i:\ u_i=u\}$. For any $Z_1$, $Z_2\in\mathcal{X}_{M,L,l}$, if $MS(Z_1)=MS(Z_2)$, their DNA-distance between $Z_1$ and $Z_2$ is defined as
$$\mathcal{D}(Z_1,Z_2):=\max\limits_{u\in S(Z_1)}\min\limits_{\substack{\pi\in{\rm{BI}}(I(u,Z_1),\\I(u,Z_2))}}\max\limits_{x\in I(u,Z_1)}d_H(x,\pi(x));$$
otherwise $\mathcal{D}(Z_1,Z_2)=\infty$. In \cite{boruchovsky2023dna}, the authors showed that this is a distance function. For a code $\mathcal{C}\subset\mathcal{X}_{M,L,l}$, the minimum DNA-distance of $\mathcal{C}$ is defined as
$$\mathcal{D}(\mathcal{C}):=\min_{Z_1\ne Z_2\in C}\mathcal{D}(Z_1,Z_2).$$
\indent In this paper, we will always equip $\mathbb{R}^2$ with the following partial order:
$$(s,t)\leq(s',t')\iff s\le s'\ \text{and}\ t\le t'.$$
Note that this is not a total order. For any $s=(ind,u)\in\mathbb{F}_2^L$, let
\begin{align*}
   & w_{H,l}(s)=w_H(ind),          \\
   & w_{H,L-l}(s)=w_H(u),          \\
   & w_{H,L}(s)=(w_H(ind),w_H(u)),
\end{align*}
where $w_H$ denotes the Hamming weight function. For any $s$, $s'\in\mathbb{F}_2^L$, let
\begin{align*}
   & d_{H,l}(s,s')=w_{H,l}(s-s'),     \\
   & d_{H,L-l}(s,s')=w_{H,L-l}(s-s'), \\
   & d_{H,L}(s,s')=w_{H,L}(s-s').
\end{align*}
Although $d_{H,L}$ is not a distance function, it satisfies the following analogous properties:
\begin{enumerate}
  \item $d_{H,L}(x,y)\ge (0,0)$ for any $x$, $y\in\mathbb{F}_2^L$ and $d_{H,L}(x,y)=(0,0)$ if and only if $x=y$;
  \item $d_{H,L}(x,y)=d_{H,L}(y,x)$ for any $x$, $y\in\mathbb{F}_2^L$;
  \item $d_{H,L}(x,z)\le d_{H,L}(x,y)+d_{H,L}(y,z)$ for any $x$, $y$, $z\in\mathbb{F}_2^L$.
\end{enumerate}

\indent The following Hall's marriage theorem will be frequently used in this section.

\begin{theorem}[\cite{diestel2005graph}, Theorem 2.1.2]
  For a finite bipartite graph $G=(L\cup R,E)$, there is an $L$-perfect matching if and only if for every subset $Y\subset L$, it holds that $|Y|\le|N_G(Y)|$, where $N_G(Y)$ is the set of all vertices that are adjacent to some vertex in $Y$.
\end{theorem}

\indent First, we consider the case where $\tau=1$.

\begin{lemma}\label{20230701lemmaone}
  Let $Z$, $Z'\in\mathcal{X}_{M,L,l}$. Then
  $$B^K_{(1,e_i,e_d)}(Z)\cap B^K_{(1,e_i,e_d)}(Z')\neq\emptyset$$
  if and only if there exists $\pi\in{\rm{BI}}(Z,Z')$ such that for any $x\in Z$, $d_{H,L}(x,\pi(x))\le(2e_i,2e_d)$.
\end{lemma}

\begin{proof}
  Assume that there exists $\pi\in{\rm{BI}}(Z,Z')$ such that for any $x\in Z$, $d_{H,L}(x,\pi(x))\le(2e_i,2e_d)$. We may assume that $Z=\{s_1=(ind_1,u_1),\cdots,s_M=(ind_M,u_M)\}$ and $Z'=\{s_1'=(ind'_1,u'_1),\cdots,s_M'=(ind'_M,u'_M)\}$ such that $\pi(s_j)=s_j'$ for any $1\le j\le M$. For any $1\le j\le M$, since $d_{H,L}(s_j,s_j')\le (2e_i,2e_d)$, i.e., $d_{H}(ind_j,ind_j')\le 2e_i$ and $d_{H}(u_j,u_j')\le 2e_d$, there exists $s_j''=(ind_j'',u_j'')\in\mathbb{F}_2^L$ such that
  \begin{align*}
    d_{H}(ind_j,ind_j'')\le e_i, & \quad d_{H}(ind_j',ind_j'')\le e_i, \\
    d_{H}(u_j,u_j'')\le e_d,     & \quad d_{H}(u_j',u_j'')\le e_d.
  \end{align*}
  Then the multiset
  $$\{\{s_1'',\cdots,s_1'',\cdots,s_{M}'',\cdots,s_{M}''\}\}$$
  in which each $s_j''$ occurs $K$ times, lies in
  $$B^K_{(1,e_i,e_d)}(Z)\cap B^K_{(1,e_i,e_d)}(Z').$$
  \indent Conversely, assume that
  $$A=\{\{s_1'',\cdots,s_{MK}''\}\}\in B^K_{(1,e_i,e_d)}(Z)\cap B^K_{(1,e_i,e_d)}(Z').$$
  Consider the bipartite graph $G=(Z\cup Z',E)$ where
  $$E=\{(x,y)\in Z\times Z':\ d_H(x,y)\le (2e_i,2e_d)\}.$$
  Assume, for a contradiction, that there does not exist $\pi\in{\rm{BI}}(Z,Z')$ satisfying the desired properties. By Hall's marriage theorem, there exists $Y\subset Z$ such that $|Y|>|N_G(Y)|$. Assume that $Y'=Z'\backslash N_G(Y)$, $|Y|=r$ and $|Y'|=s$. Then $s>M-r$. From the definition of $A$, there is a multi-subset $A_1$ of $A$ having $rK$ elements (counting multiplicity) such that for any $x\in A_1$, there exists  $y\in Y$ such that $d_{H,L}(x,y)\le (e_i,e_d)$ and there is a multi-subset $B_1$ of $A$ having $sK$ elements (counting multiplicity) such that for any $z\in B_1$, there exists $y'\in Y'$ such that $d_{H,L}(z,y')\le (e_i,e_d)$. If there exists $x\in A_1\cap B_1$, then there exist elements $y\in Y$ and $y'\in Y'$ such that $d_{H,L}(x,y)\le (e_i,e_d)$ and $d_{H,L}(x,y')\le (e_i,e_d)$. It follows that
  $$d_{H,L}(y,y')\le d_{H,L}(x,y)+d_{H,L}(x,y')\le (2e_i,2e_d),$$
  which implies that $y'\in N_G(Y)$. This is a contradiction and thus $A_1\cap B_1=\emptyset$. Since $|A_1|=rK$ and $|B_1|=sK$, we have $rK+sK\le MK$, i.e., $s\le M-r$. This contradicts the fact that $s>M-r$. Therefore, there exists $\pi\in{\rm{BI}}(Z,Z')$ such that for any $x\in Z$, $d_{H,L}(x,\pi(x))\le(2e_i,2e_d)$.
\end{proof}

\indent The following theorem follows immediately from Lemma \ref{20230701lemmaone}.

\begin{theorem}\label{20230701theoremone}
  A code $\mathcal{C}\subset\mathcal{X}_{M,L,l}$ is a $(1,e_i,e_d)_K$-DNA-correcting code if and only if for any $Z$, $Z'\in \mathcal{C}$ with $Z\neq Z'$, it holds that for any $\pi\in{\rm{BI}}(Z,Z')$, there exists $x\in Z$ such that $d_{H,L}(x,\pi(x))\not\le(2e_i,2e_d)$, i.e., $d_{H,l}(x,\pi(x))>2e_i$ or $d_{H,L-l}(x,\pi(x))>2e_d$.
\end{theorem}

\indent Let us consider the special case where $e_d=0$.

\begin{lemma}\label{20230701lemmatwo}
  Let $\mathcal{C}\subset\mathcal{X}_{M,L,l}$ and let $r\in\mathbb{R}_{\ge 0}$. Then $\mathcal{D}(\mathcal{C})\le r$ if and only if there exists $Z$, $Z'\in\mathcal{C}$ with $Z\ne Z'$ such that there exists $\pi\in{\rm{BI}}(Z,Z')$ satisfying that $d_{H,L}(x,\pi(x))\le (r,0)$ for any $x\in Z$.
\end{lemma}

\begin{proof}
  Assume that $\mathcal{D}(\mathcal{C})\le r$. Then there exists $Z$, $Z'\in\mathcal{C}$ with $Z\ne Z'$ such that $\mathcal{D}(Z,Z')\le r$. By the definition of the DNA-distance, we must have $MS(Z)=MS(Z')$ and for any $u\in S(Z)$, there exists $\pi_u\in{\rm{BI}}(I(u,Z),I(u,Z'))$ such that for any $ind\in I(u,Z)$, $d_H(ind,\pi_u(ind))\le r$. We can glue these $\pi_u$ ($u\in S(Z)$) into a bijection $\pi:Z\rightarrow Z'$. It is clear that for any $x=(ind,u)\in Z$, $d_{H,l}(x,\pi(x))=d_H(ind,\pi_u(ind))\le r$ and $d_{H,L-l}(x,\pi(x))=0$, i.e., $d_{H,L}(x,\pi(x))\le (r,0)$.\\
  \indent Conversely, assume that there exists $Z$, $Z'\in\mathcal{C}$ with $Z\ne Z'$ such that there exists $\pi\in{\rm{BI}}(Z,Z')$ satisfying that $d_{H,L}(x,\pi(x))\le (r,0)$ for any $x\in Z$. Then $MS(Z)=MS(Z')$ and for any $u\in S(Z)$, $\pi$ induces a bijection $\pi_u\in{\rm{BI}}(I(u,Z),I(u,Z'))$ such that for any $ind\in I(u,Z)$,
  $$d_H(ind,\pi_u(ind))=d_{H,l}((ind,u),\pi((ind,u)))\le r.$$
  It follows that $\mathcal{D}(Z,Z')\le r$ and thus $\mathcal{D}(\mathcal{C})\le r$.
\end{proof}

\indent The following result can now be easily derived from Theorem \ref{20230701theoremone} and Lemma \ref{20230701lemmatwo}.

\begin{corollary}[\cite{boruchovsky2023dna}, Theorem 2]
  A code $\mathcal{C}\subset\mathcal{X}_{M,L,l}$ is a $(1,e_i,0)_K$-DNA-correcting code if and only if $\mathcal{D}(\mathcal{C})>2e_i$.
\end{corollary}

\indent Next, we consider the case where $\frac{K}{2}\le\lfloor\tau K\rfloor<K$.

\begin{lemma}\label{20230701lemmathree}
  Assume that $\frac{K}{2}\le\lfloor\tau K\rfloor<K$ and let $Z$, $Z'\in\mathcal{X}_{M,L,l}$. If there exists $\pi\in{\rm{BI}}(Z,Z')$ such that for any $x\in Z$, $d_{H,L}(x,\pi(x))\le(e_i,e_d)$, then
  $$B^K_{(\tau,e_i,e_d)}(Z)\cap B^K_{(\tau,e_i,e_d)}(Z')\neq\emptyset.$$
\end{lemma}
\begin{proof}
  We may assume that
  $$Z=\{s_1=(ind_1,u_1),\cdots,s_M=(ind_M,u_M)\}$$ and
  $$Z'=\{s_1'=(ind'_1,u'_1),\cdots,s_M'=(ind'_M,u'_M)\}$$
  such that $\pi(s_j)=s_j'$ for any $1\le j\le M$. Consider the multiset
  \begin{align*}
    A=\{\{ & s_1,\cdots,s_1,s_1',\cdots,s_1',\cdots,s_M,\cdots,s_M, \\
           & s_M',\cdots,s_{M}'\}\},
  \end{align*}
  in which each $s_j$ occurs $\left\lfloor\tau K\right\rfloor$ times and each $s_j'$ occurs $K-\left\lfloor\tau K\right\rfloor$ times for any $1\le j\le M$. Since $\frac{K}{2}\le\lfloor\tau K\rfloor$ and for any $1\le j\le M$, $d_{H}(ind_i,ind'_i)\le e_i$ and $d_{H}(u_i,u'_i)\le e_d$ we have $A\in B^K_{(\tau,e_i,e_d)}(Z)\cap B^K_{(\tau,e_i,e_d)}(Z')$.
\end{proof}

\indent The following theorem follows immediately from Lemma \ref{20230701lemmathree}.

% 从这里开始

\begin{theorem}\label{20230701theoremtwo}
  Assume that $\frac{K}{2}\le\lfloor\tau K\rfloor<K$. If $\mathcal{C}\subset\mathcal{X}_{M,L,l}$ is a $(\tau,e_i,e_d)_K$-DNA-correcting code, then for any $Z$, $Z'\in\mathcal{C}$ with $Z\neq Z'$, it holds that for any $\pi\in{\rm{BI}}(Z,Z')$, there exists $x\in Z$ such that $d_{H,L}(x,\pi(x))\not\le(e_i,e_d)$.
\end{theorem}

\indent The following result can now be easily derived from Theorem \ref{20230701theoremtwo} and Lemma \ref{20230701lemmatwo}.

\begin{corollary}[\cite{boruchovsky2023dna}, Lemma 2]
  Assume that $\frac{K}{2}\le\lfloor\tau K\rfloor<K$. If $\mathcal{C}\subset\mathcal{X}_{M,L,l}$ is a $(\tau,e_i,0)_K$-DNA-correcting code, then $\mathcal{D}(\mathcal{C})>e_i$.
\end{corollary}

\indent For any $r_1$, $r_2\in\mathbb{R}_{\ge 0}$, put
\begin{align*}
  \overline{\mathcal{X}}_{M,L,l}^{(r_1,r_2)}=\{  Z\in & \mathcal{X}_{M,L,l}:\ \forall\ x\ne y\in Z, \\
                                                      & d_{H,L}(x,y)\not\le(r_1,r_2)\}.
\end{align*}

Restricting ourselves to $\overline{\mathcal{X}}_{M,L,l}^{(2e_i,2e_d)}$, the converse of Lemma \ref{20230701lemmathree} is also true.

\begin{lemma}\label{20230702lemmaone}
  Assume that $\frac{K}{2}\le\lfloor\tau K\rfloor<K$ and let $Z$, $Z'\in\overline{\mathcal{X}}_{M,L,l}^{(2e_i,2e_d)}$. Then
  $$B^K_{(\tau,e_i,e_d)}(Z)\cap B^K_{(\tau,e_i,e_d)}(Z')\neq\emptyset$$
  if and only if there exists $\pi\in{\rm{BI}}(Z,Z')$ such that for any $x\in Z$, $d_{H,L}(x,\pi(x))\le(e_i,e_d)$.
\end{lemma}

\begin{proof}
  By Lemma \ref{20230701lemmathree}, we only need to prove the necessity. Take $A\in B^K_{(\tau,e_i,e_d)}(Z)\cap B^K_{(\tau,e_i,e_d)}(Z')$. Since $\tau<1$, for any $x\in Z$, $x$ lies in $A$, which implies that there exists $x'\in Z'$ such that $d_{H,L}(x,x')\le(e_i,e_d)$. Assume that there exists another element $x''\in Z'$ with $x''\ne x'$ such that $d_{H,L}(x,x'')\le (e_i,e_d)$ for some $x'\in Z'$. Then
  $$d_{H,L}(x',x'')\le d_{H,L}(x,x')+d_{H,L}(x,x'')\le (2e_i,2e_d),$$
  which contradicts the hypothesis that $Z'\in\overline{\mathcal{X}}_{M,L,l}^{(2e_i,2e_d)}$. Therefore, we have a map $\pi:Z\rightarrow Z'$ given by $x\mapsto x'$, which satisfies that $d_{H,L}(x,\pi(x))\le (e_i,e_d)$ for any $x\in Z$. Using the hypothesis that $Z\in\overline{\mathcal{X}}_{M,L,l}^{(2e_i,2e_d)}$, we can show that $\pi$ is a bijection as above.
\end{proof}

\indent If $\frac{K}{2}\le\lfloor\tau K\rfloor<\frac{MK}{2M-1}$, then restricting ourselves to $\overline{\mathcal{X}}_{M,L,l}^{(e_i,e_d)}$ is enough.

\begin{lemma}\label{20230702lemmatwo}
  Assume that $\frac{K}{2}\le\lfloor\tau K\rfloor<\frac{MK}{2M-1}$ and let $Z$, $Z'\in\overline{\mathcal{X}}_{M,L,l}^{(e_i,e_d)}$. Then
  $$B^K_{(\tau,e_i,e_d)}(Z)\cap B^K_{(\tau,e_i,e_d)}(Z')\neq\emptyset$$
  if and only if there exists $\pi\in{\rm{BI}}(Z,Z')$ such that for any $x\in Z$, $d_{H,L}(x,\pi(x))\le(e_i,e_d)$.
\end{lemma}
\begin{proof}
  By Lemma \ref{20230701lemmathree}, we only need to prove the necessity. Assume that there does not exist any bijection $\pi\in{\rm{BI}}(Z,Z')$ satisfying the desired condition. Consider the bipartite graph $G=(Z\cup Z',E)$ where
  $$E=\{(x,x')\in Z\times Z':\ d_{H,L}(x,x')\le (e_i,e_d)\}.$$
  Then the assumption is equivalent to saying that there does not exist any $Z$-perfect matching. By Hall's marriage theorem, there exists $Y\subset Z$ such that $|Y|>|N_G(Y)|$. We may assume that $Y$ has the smallest cardinality among all subsets of $Z$ satisfying this condition. Assume that $|Y|=r$ and $|N_G(Y)|=s$. Since $\tau<1$, for any $x\in Z$, there exists $x'\in Z'$ such that $(x,x')\in E$, which implies that $s\ge 1$ and thus $r\ge 2$.\\
  \indent For any $y'\in N_G(Y)$, take any $y\in Y$ such that $(y,y')\in E$ and put $Y_0=Y\backslash\{y\}$. If $y'\not\in N_G(Y_0)$, then
  $$|N_G(Y_0)|\le |N_G(Y)|-1<|Y|-1=|Y_0|,$$
  which contradicts the minimality of $Y$. Hence there exists $y_0\in Y_0$ such that $(y_0,y')\in E$. Since $Z\in\overline{\mathcal{X}}_{M,L,l}^{(e_i,e_d)}$, we have $d_{H,L}(y,y_0)\not\le(e_i,e_d)$. Since $(y_0,y')\in E$, i.e., $d_{H,L}(y_0,y')\le (e_i,e_d)$, we must have $y\ne y'$.\\
  \indent Take
  $$A\in B^K_{(\tau,e_i,e_d)}(Z)\cap B^K_{(\tau,e_i,e_d)}(Z').$$
  Since $\tau<1$, each strand in $Y$ occurs at least $K-\lfloor\tau K\rfloor$ times in $A$. Such strands can occur only as noisy copies of strands in $N_G(Y)$. Furthermore, by the discussions in the previous paragraph, they can occur only as error copies of strands in $N_G(Y)$. It follows that
  $$r(K-\lfloor\tau K\rfloor)\le s\lfloor\tau K\rfloor ,$$
  which implies that
  $$rK\le (s+r)\lfloor\tau K\rfloor\le (2r-1)\lfloor\tau K\rfloor$$
  and thus
  $$\lfloor\tau K\rfloor\ge\frac{rK}{2r-1}\ge\frac{MK}{2M-1}.$$
  This contradicts our hypothesis and thus
  $$B^K_{(\tau,e_i,e_d)}(Z)\cap B^K_{(\tau,e_i,e_d)}(Z')=\emptyset.$$
\end{proof}

\begin{remark}
  Distinct from the other results in this section, this lemma is not a formal generalization of any result in \cite{boruchovsky2023dna}.
\end{remark}

\indent The following theorem follows immediately from Lemma \ref{20230702lemmaone} and Lemma \ref{20230702lemmatwo}.

\begin{theorem}\label{20230701themfour}
  Assume that $\frac{K}{2}\le\lfloor\tau K\rfloor<K$.
  \begin{enumerate}
    \item For any code $\mathcal{C}\subset\overline{\mathcal{X}}_{M,L,l}^{(2e_i,2e_d)}$, $\mathcal{C}$ is a $(\tau,e_i,e_d)_K$-DNA-correcting code if and only if for any $Z$, $Z'\in\mathcal{C}$ with $Z\neq Z'$, it holds that for any $\pi\in{\rm{BI}}(Z,Z')$, there exists $x\in Z$ such that $d_{H,L}(x,\pi(x))\not\le(e_i,e_d)$.
    \item If moreover $\lfloor\tau M\rfloor<\frac{MK}{2M-1}$, then for any code $\mathcal{C}\subset\overline{\mathcal{X}}_{M,L,l}^{(e_i,e_d)}$, $\mathcal{C}$ is a $(\tau,e_i,e_d)_K$-DNA-correcting code if and only if for any $Z$, $Z'\in\mathcal{C}$ with $Z\neq Z'$, it holds that for any $\pi\in{\rm{BI}}(Z,Z')$, there exists $x\in Z$ such that $d_{H,L}(x,\pi(x))\not\le(e_i,e_d)$.
  \end{enumerate}
\end{theorem}

\indent In \cite{boruchovsky2023dna}, the authors considered the following set
$$\overline{\mathcal{X}}_{M,L,l}=\{Z\in\mathcal{X}_{M,L,l}:\ |S(Z)|=M\}.$$
It is clear that $\overline{\mathcal{X}}_{M,L,l}\subset\overline{\mathcal{X}}_{M,L,l}^{(2e_i,0)}$ and thus the following result follows immediately from Theorem \ref{20230701themfour} and Lemma \ref{20230701lemmatwo}.

\begin{corollary}[\cite{boruchovsky2023dna}, Corollary 2]
  Assume that $\frac{K}{2}\le\lfloor\tau K\rfloor<K$. For any code $\mathcal{C}\subset\overline{\mathcal{X}}_{M,L,l}$, $\mathcal{C}$ is a $(\tau,e_i,0)_K$-DNA-correcting code if and only if $\mathcal{D}(\mathcal{C})>e_i$.
\end{corollary}

\begin{comment}

\begin{lemma}
  Assume that $\lfloor\tau K\rfloor<\frac{K}{2}$ and let $Z$, $Z'\in\mathcal{X}_{M,L,l}$. If
  $$B^K_{(\tau,e_i,e_d)}(Z)\cap B^K_{(\tau,e_i,e_d)}(Z')\ne\varnothing,$$
  then
  \begin{enumerate}
    \item for any $x\in Z$ and any $x'\in Z'$ with $d_{H,L}(x,x')\le(e_i,e_d)$, either $x=x'$ or $d_{H,L}(y,x')\le (e_i,e_d)$ for some $y\in Z$ distinct from $x$, and
    \item there exists $\pi\in{\rm{BI}}(Z,Z')$ such that for any $x\in Z$, $d_{H,L}(x,\pi(x))\le(e_i,e_d)$.
  \end{enumerate}
\end{lemma}

\begin{proof}
  Let $G=(Z\cup Z',E)$ be the graph defined in the proof of Lemma \ref{20230702lemmatwo}. Assume, for a contradiction, that there does not exist $\pi\in{\rm{BI}}(Z.Z')$ of the desired property. Then by Hall's marriage theorem, there exists $Y\subset Z$ such that $|Y|>|N_G(Y)|$.
\end{proof}

\end{comment}

\bibliographystyle{IEEEtran}
\bibliography{references}{}

\end{document}